\documentclass[a4paper,11pt]{article}

\usepackage{amscd,amssymb,amsthm}
\usepackage{palatino}
\usepackage{color}

\newtheorem{theorem}{Theorem}

\newtheorem{lemma}{Lemma}
\newtheorem{definition}{Definition}

\newcommand{\sty}{\displaystyle}

\begin{document}

\title{Involutions, Trace Maps, and \\ Pseudorandom Numbers }

\author{Michele Elia\thanks{Politecnico di Torino, Italy}, ~~
 Davide Schipani  \thanks{University of Zurich, Switzerland}
}

%\date{August 2016}
\maketitle

\thispagestyle{empty}

\begin{abstract}
\noindent 
Interesting properties of the partitions of  a finite field $\mathbb F_q$ induced by the combination of involutions and trace maps are studied. 
The special features of  involutions of the form  $\frac{u}{z}$, $u$ being a fixed element of $\mathbb F_q$, are exploited to generate  pseudorandom numbers, the randomness resting on the uniform distribution of the images of zero-trace elements among the sets of non-zero trace elements of  $\mathbb F_q$.
\end{abstract}


\vspace{2mm}

% ********************************************************************

\vspace{8mm}

\section{Introduction}
An interesting problem in studying the partition of a finite field $\mathbb F_q$, $q=p^m$ with $p$ prime and $m>1$, in subsets of  elements with the same trace, is the study of the partition of the images of these sets via a bijection $f$ of $\mathbb F_q$ into itself by means of the trace function. \\
Let $\mathcal A_h$ be the subset of elements of trace $h \in \mathbb F_p$, and let $\mathcal B_{hk}$ be the subset of $f(\mathcal A_h)$ consisting of all elements of trace $k$. The problem is to compute the cardinalities 
$B_{hk}=|\mathcal B_{hk}|$ which form a square array. 
Using indicator functions and properties of the Kloosterman sums it will be proved that, when $f$ belongs to a certain class of functions,  $ B_{h0}=B$ holds independently of $h >0$. 
The same result can be generalized using different techniques when the base field is not a prime field. 
The next section proves these results, and the following section is devoted to examples. Lastly, a section of potential applications to pseudorandom number generation  concludes the paper.

% VEDI Estensioni dire il contenuto delle sezioni
% ********************************************************************
\section{Results}
Let $Tr_w(\alpha)$ and $N_w(\alpha)$ be the trace function and the norm function, respectively \cite{lidl,mceliece}, defined over $\mathbb F_{w^s}$, i.e.
$$  Tr_w(\alpha)= \sum_{i=0}^{s-1}  \alpha^{w^i}  ~~~~~~~,~~ ~~~N_w(\alpha)= \prod_{i=0}^{s-1}  \alpha^{w^i}  ~~,~~~~ ~~ \forall \alpha \in \mathbb F_{w^s}~~. $$

\subsection{Trace partition of $\mathbb F_{p^{m}}$ into $\mathbb F_{p}$}
The cardinality $B_{hk}$ of the sets $\mathcal B_{hk}$ form a $p\times p$ square array $\mathbf B$ with indeces $h,k=0, \ldots, p-1$, and  can be written in terms of an indicator function $\delta_p$. Let $\zeta_p$ denote a $p$-th root of unity.  An additive $p$-character over $\mathbb F_q$ is defined as
$$   \chi_p(\alpha) = e( Tr_p(\alpha)) = \zeta_p^{ Tr_p(\alpha)}  ~~~~\forall~\alpha \in \mathbb F_q~~, $$ 
 then an indicator function for the elements of trace $0$  is
$$  \delta_p(\alpha) =\frac{1}{p}  \sum_{j=0}^{p-1}  \chi_p(\alpha)^j = \frac{1}{p}  \sum_{j=0}^{p-1}  \chi_p(j\alpha) = \left\{  \begin{array}{lcl} 
                1 &  \mbox{if} &    Tr_p(\alpha)=0    \\
                0 &      &  \mbox{otherwise}   \\
        \end{array}  \right.
$$
Letting $\beta_k$ be an element of $\mathbb F_q$ with trace $k$, and using the indicator function $\delta_p$, the cardinality  of  $\mathcal B_{hk}$  is
\begin{equation}
  \label{symm}
 B_{hk}  = \sum_{\alpha \in \mathbb F_q}    \delta_p(\alpha-\beta_h)   \delta_p(f(\alpha)-\beta_k)    ~~. 
\end{equation}

\noindent
By means of this equation  the following property of involutions  \cite{charpin} can easily be shown.
\begin{lemma}
   \label{lem1}
Assuming that $f$ is an involution, the array  $\mathbf B$  of cardinalities
% of the sets $\mathcal B_{hk}$ are 
is symmetric in the indices $h$ and $k$, i.e.
$$  B_{hk} =   B_{kh}   ~~. $$
\end{lemma}

\begin{proof}
Performing the substitution $\beta=f(\alpha)$ in equation (\ref{symm}), due to the involution $f$
 we have $\alpha= f(\beta)$, hence
$$ B_{hk} = \sum_{\beta \in \mathbb F_q}  \delta_p(f(\beta)-\beta_h) \delta_p(\beta-\beta_k),  $$
 where, by definition, the summation      on the right-hand side is $B_{kh}$.
\end{proof}

\vspace{3mm}
\noindent
The set of cardinalities $B_{hk}$ 
%form a square array, and
 can be viewed as double discrete Fourier transforms of
a sequence of Kloosterman sums. We have
$$ B_{hk}  = \sum_{\alpha \in \mathbb F_q}    \delta_p(\alpha-\beta_h)   \delta_p(f(\alpha)-\beta_k) = \frac{1}{p^2} \sum_{j=0}^{p-1} \sum_{\ell=0}^{p-1}   \sum_{\alpha \in \mathbb F_q}    e(j Tr_p(\alpha-\beta_h))   e(\ell  Tr_p(f(\alpha)-\beta_k))  ~~. $$
Since the trace operator is linear, this equation can be rewritten as
$$ B_{hk}  = \frac{1}{p^2} \sum_{j=0}^{p-1} \sum_{\ell=0}^{p-1}    \sum_{\alpha \in \mathbb F_q}    e(j Tr_p(\alpha) -j h)   e(\ell  Tr_p(f(\alpha))-\ell k)
      =\frac{1}{p^2} \sum_{j=0}^{p-1} \sum_{\ell=0}^{p-1} e(-j h-\ell k)   \sum_{\alpha \in \mathbb F_q}    e( Tr_p(j\alpha+ \ell f(\alpha)))  ~~.  $$
Considering the bijection $f$ of $\mathbb F_q$ into itself defined as
$$      f(\alpha) = \left\{ \begin{array}{lcl}
                               \sty        \frac{u}{\alpha} & ~~ &  \mbox{if} ~ \alpha \neq 0    \\
                                                              & ~~ &    \\
                                     0                       & ~~ &  \mbox{if} ~ \alpha = 0  
                                    \end{array}   \right.
$$
where $u$ is a fixed non-zero element of $\mathbb F_q$, the summation over $\alpha$ can be reduced to a Kloosterman sum
$$    \sum_{\alpha \in \mathbb F_q}    e( Tr_p(j\alpha+ \ell f(\alpha)))  =  \sum_{\alpha \in \mathbb F_q^*}    e( Tr_p(j\alpha+ \ell f(\alpha)))   +  1   = 1 + \sum_{\alpha \in \mathbb F_q^*}  \chi_p\left(j\alpha + \ell \frac{u}{\alpha}\right) ~~. $$  
Some useful properties of Kloosterman sums are recalled:
\begin{enumerate}
  \item  A Kloosterman sum $\mathbf K(\chi_p,a,b)$ over a finite field $\mathbb F_q$ is defined as  
$$ \mathbf K(\chi_p,a,b)= \sum_{\alpha \in \mathbb F_q^*}  \chi_p\left(a\alpha +  \frac{b}{\alpha}\right) ~~a,b \in  \mathbb F_q $$
  \item  $\mathbf K(\chi_p,a,b)$  is a real number.
  \item  $\mathbf K(\chi_p,a,b)$  depends only on the product $ab$ and character $\chi_p$ \cite{lidl}, precisely
 $$ \mathbf K(\chi_p,a,b)= \left\{  \begin{array}{lcl}
                 q-1 & ~~ &  \mbox{if} ~ ~ a=b= 0    \\       
                 -1 & ~~ &  \mbox{if} ~ ~ a \neq 0 ~\mbox{and} ~b= 0, ~~\mbox{or}  ~~ a = 0 ~\mbox{and} ~b\neq 0  \\  
               \mathbf K(\chi_p,ab) &  &  \mbox{if} ~ ~ ab\neq 0
      \end{array}  \right.
$$
   \item The sum satisfies the following inequality (Weil)
$$    | \mathbf K(\chi_p,ab)|  < 2 \sqrt{q}       ~~~ab \neq 0~~~.  $$  
\end{enumerate}

\begin{theorem}
   \label{lem2}
Assume as above that $f$ is the involution $\frac{u}{\alpha}$, $u\neq 0$.
The elements  $B_{h0}$ of the first column of $\mathbf B$  are equal for every $h>0$. 
Every row of the square array $B_{hk}$ with $h\geq 1$, is a permutation of the row indexed by $h=1$, and by symmetry the same occurs for the columns. 
\end{theorem}

%\end{document}
\begin{proof}
The value  $B_{hk}$ can be expressed in terms of Kloosterman sums as follows
$$  B_{hk}  = \frac{1}{p^2} \sum_{j=0}^{p-1} \sum_{\ell=0}^{p-1}    e(-jh -\ell k)    \left[1+ \sum_{\alpha \in \mathbb F^*}    e\left( Tr_p\left(j\alpha+\ell \frac{u}{\alpha}\right)\right)   \right]  ~~. $$
The values of the Kloosterman sum are different depending on whether or not $j$ and $\ell$ are zero. In particular, if $j=\ell=0$ its value is $q-1$; if only one of $j$ and $\ell$ is zero the value is $-1$; otherwise it depends only on the product $j\ell$ modulo $p$. Therefore, the double sum can be split into four terms

\begin{equation}
   \label{geneq}
 B_{hk}  = \frac{1}{p^2}  \left[ q -   \sum_{j=1}^{p-1} 0\cdot  e(-jh)  -   \sum_{\ell=1}^{p-1} 0 \cdot  e(-\ell k) +  \sum_{j=1}^{p-1} \sum_{\ell=1}^{p-1}   e(-jh -\ell k)   \{1+K(\chi_p, j\ell u)\}   \right]  ~~. 
\end{equation}

\noindent Taking $k=0$ in equation (\ref{geneq}) we get
$$
  B_{h0}  = \frac{1}{p^2}  \left[ q-(p-1)+  \sum_{j=1}^{p-1} e(-jh)  \sum_{\ell=1}^{p-1}     K(\chi_p, j\ell u)   \right]  ~~h \neq 0. 
$$
Noting that the product $j\ell$ in the function $K(\chi_p, j\ell u)$ is performed in $\mathbb F_q$, that is $t=j\ell$ is taken modulo $p$, it follows that, for a fixed $j$, as $\ell$ runs over $\mathbb F_p^*$, the index $t$ runs over the same range; therefore the summation  $\sum_{t=1}^{p-1}   K(\chi_p, t u) $ is independent of $j$, 
which implies that  $  B_{h0}$ is independent  of $h$, because the summation
$ \sum_{j=1}^{p-1} e(-jh) =-1$. In conclusion
\begin{equation}
   \label{zeroeq}
  B_{h0}  = \frac{1}{p^2}  \left[ q-(p-1)-  \sum_{t=1}^{p-1}     K(\chi_p, t u)   \right]  ~~h \neq 0. 
\end{equation}

\noindent When $hk\neq 0$ in equation (\ref{geneq}) we get
$$
\frac{1}{p^2}  \left[ q + 1+  \sum_{j=1}^{p-1} \sum_{\ell=1}^{p-1}   e(-jh -\ell k)    K(\chi_p, j\ell u) \right] ~~.
$$

\noindent Set $t=j \ell$,  and solving for $j$ we obtain $j=t \bar \ell$, with $\bar \ell$ denoting the inverse of $\ell$ modulo $p$,   the summation over $j$ is changed into a summation over $t$  
$$ \frac{1}{p^2}  \left[ q + 1+  \sum_{t=1}^{p-1} \sum_{\ell=1}^{p-1}   e(-t\bar \ell h -\ell k)    K(\chi_p, t u) \right]   $$
which shows that the summation over $\ell$ is a Kloosterman sum $K(\chi_o, thk)$, with $\chi_o(\gamma) =e(\gamma), ~\gamma \in \mathbb F_p$ ; we can thus write the expression as
$$ \frac{1}{p^2}  \left[ q + 1+  \sum_{t=1}^{p-1} K(\chi_o, thk)   K(\chi_p, t u) \right]   $$
If we fix $h$ and let $k$ vary from $1$ to $p-1$ we see that every row with $h\geq 1$, is a permutation of the row indexed by $h=1$, since all possible values for $hk$ in $1,\dots,p-1$ are obtained; clearly the same occurs for the columns. 

\end{proof}

\subsection{Trace partition of $\mathbb F_{p^{mr}}$ into $\mathbb F_{p^m}$}
The partition property also holds when we consider a field $\mathbb F_{p^{mr}}$ that is an extension of
$\mathbb F_{p^m}$ with $m >1$, but in this case the proof cannot be developed by arguing as above. In particular, a new indicator function that avoids using Kloosterman sums will be introduced. However, this
 will prevent bounds being derived by exploiting properties of the Kloosterman sums, as shown in the following section.

\begin{definition}
    \label{def1}
Let $\delta_q(\alpha)$ be an indicator function of the subset of the $0$-trace  elements of  $\mathbb F_{q^r}$, i.e.
$$  \delta_q(\alpha) =  1- N_p(Tr_q(\alpha))^{p-1} ~  \bmod p~ = \left\{   \begin{array}{lcl} 
                1 &  \mbox{if} &   Tr_q(\alpha)=0    \\
                0 &  \mbox{if} &   Tr_q(\alpha) \neq 0    \\
        \end{array}  \right.   ~~.
$$
\end{definition}

\noindent
This definition is a straightforward consequence of the fact that
  $Tr_q(\alpha)$ is an element of $\mathbb F_q$,  $N_p(Tr_q(\alpha))$ is an element of
$\mathbb F_p$, and  $ N_p(Tr_q(\alpha))^{p-1}$ is $0$ if and only if $Tr_q(\alpha)=0$, and  $1$ otherwise.

\noindent
Using  this indicator function, we can count the number of elements of $\mathcal B_{hk}$ as
\begin{equation}
  \label{symm2}
 B_{hk}  = \sum_{\alpha \in \mathbb F_{q^r}}    \delta_q(\alpha-\beta_h)   \delta_q(f(\alpha)-\beta_k)    ~~. 
\end{equation}
where $\beta_j \in \mathbb F_{q^r}$ is such that $Tr_q(\beta_j)= \omega_j \in \mathbb F_q$.

\begin{theorem}
  \label{theo1}
The array $B_{hk}$ of cardinalities satisfies the following properties:
\begin{itemize}
   \item $B_{hk}=B_{kh}$ if $f$ is an involution.
   \item  $B_{h0}$ is independent of $h\neq 0$ if $f=\frac{u}{z}$, with $u \in \mathbb F_{q^r}$.
   \item Every row of $B_{hk}$, with $h\geq 1$, %and $k \geq 1$, 
   is a permutation of the row $B_{1k}$ with index $1$  if $f=\frac{u}{z}$, with $u \in \mathbb F_{q^r}$.
\end{itemize}
\end{theorem}

\begin{proof}
The first property is an immediate consequence of  equation (\ref{symm2}) and the involution $f$. The change of variable $\beta=f(\alpha)$, which implies $\alpha =f(\beta)$, gives
$$  B_{hk}  = \sum_{\beta \in \mathbb F_{q^r}}    \delta_q(f(\beta)-\beta_h)   \delta_q(\beta-\beta_k)  =B_{kh}   ~~. $$
The  second property is proved writing equation (\ref{symm2}) using the given definition of $\delta_q$
$$  B_{0h}  = \sum_{\alpha \in \mathbb F_{q^r}}    \delta_q(\alpha)   \delta_q(f(\alpha)-\beta_h)   =   \sum_{\alpha \in \mathbb F_{q^r}}   [1-N_p(Tr_q(\alpha))^{p-1}] [1- N_p(Tr_q(f(\alpha))-\omega_h)^{p-1} ] ~~. $$
It is remarked that in this equation the evaluation of the powers of the norms is made modulo $p$, while 
the summation treats the resulting numbers, i.e. $0$ or $1$, as integers.
Expanding the argument of the second summation we see that it is a sum of four summations, namely
\begin{itemize}
   \item  $  \sum_{\alpha \in \mathbb F_{q^r}}  1 = q^r$
   \item  $ \sum_{\alpha \in \mathbb F_{q^r}} N_p(Tr_q(\alpha))^{p-1} = (q-1)~q^{r-1}   $, since the argument of the summation is $1$ whenever the trace of $\alpha$ is not zero, and the number of elements of $\mathbb F_{q^r}$ with non zero trace are just $q^r-q^{r-1}$.
    \item  $ \sum_{\alpha \in \mathbb F_{q^r}} N_p(Tr_q(f(\alpha))-\omega_h)^{p-1} =
                   \sum_{\beta \in \mathbb F_{q^r}} N_p(Tr_q(\beta)-\omega_h)^{p-1}   = (q-1) ~q^{r-1}   $.
    The conclusion is motivated by an argument similar to that in the previous point, simply noting that the number of elements of $\mathbb F_{q^r}$ with trace different from $\omega_h$ are just $q^r-q^{r-1}$.
    \item  $\sum_{\alpha \in \mathbb F_{q^r}}   N_p(Tr_q(\alpha))^{p-1} N_p(Tr_q(\frac{u}{\alpha})-\omega_h)^{p-1} $, since $\omega_h \neq 0$ we can collect $\omega_h$ and write
$$  \sum_{\alpha \in \mathbb F_{q^r}}   N_p(\omega_h Tr_q(\alpha))^{p-1} N_p\left( \frac{1}{\omega_h} Tr_q\left(\frac{u}{\alpha}\right)-1\right)^{p-1} = \sum_{\alpha \in \mathbb F_{q^r}}   N_p(
Tr_q(\omega_h\alpha))^{p-1} N_p\left(  Tr_q\left(\frac{u}{\omega_h\alpha}\right)-1\right)^{p-1}      $$
Now, setting $\beta=\omega_h\alpha$, it is seen that,
 as $\alpha$ runs through $\mathbb F_{q^r}$, $\beta$ does the same, and consequently
the summation is independent of $h$.
\end{itemize}
In conclusion,  $B_{0h}$ is independent of $h$, and setting $b_4= \sum_{\beta \in \mathbb F_{q^r}}   N_p(
Tr_q(\beta))^{p-1} N_p(  Tr_q(\frac{u}{\beta})-1)^{p-1}  $, we have
\begin{equation}
    \label{zeroeq1}
    B_{0h}  = q^r-2(q^r-q^{r-1}) +b_4 = 2q^{r-1}-q^r+b_4 ~~,~~     ~~1\leq h \leq q-1 ~~.  
\end{equation}

\vspace{2mm}
\noindent
The third property is proved by writing equation (\ref{symm2}) using the definition given of $\delta_q$
$$  B_{hk}  = \sum_{\alpha \in \mathbb F_{q^r}}    \delta_q(\alpha-\beta_h)   \delta_q(f(\alpha)-\beta_k)   =   \sum_{\alpha \in \mathbb F_{q^r}}   \left[1-N_p(Tr_q(\alpha)-\omega_h)^{p-1}\right] \left[1- N_p\left(Tr_q\left(\frac{u}{\alpha}\right)-\omega_k\right)^{p-1} \right] . $$
%It is remarked that in this equation the evaluation of the powers of the norms is made modulo $p$, while the summation treats the resulting numbers, i.e. $0$ or $1$, as integers. Further, 
Considering $hk\neq 0$, which %Since in this equation 
implies $\omega_h \neq 0$ and $\omega_k \neq 0$, the expression can be written as
$$  \sum_{\alpha \in \mathbb F_{q^r}}  \left[1- N_p\left( Tr_q\left(\frac{\alpha}{\omega_h}\right)-1\right)^{p-1}\right]  \left[ 1-N_p\left( Tr_q\left(\frac{u}{\omega_k \alpha}\right)-1\right)^{p-1} \right] ~~ $$
since $N_p(\omega)^{p-1}=1 \bmod p$, and $\omega Tr_q(\alpha) = Tr_q(\omega \alpha)$ if $\omega \in \mathbb F_{q^*}$.
Now, setting $\beta=\frac{\alpha}{\omega_h}$, it is seen that,
 as $\alpha$ runs through $\mathbb F_{q^r}$, $\beta$ does the same, then  the summation
$$  \sum_{\beta \in \mathbb F_{q^r}}  \left[1- N_p\left( Tr_q(\beta)-1\right)^{p-1}\right]  \left[ 1-N_p\left( Tr_q\left(\frac{u}{\omega_k\omega_h \beta}\right)-1\right)^{p-1} \right] ~~ $$
depends only on the product $\omega_k\omega_h$, thus letting $\omega_h$ to be fixed, as  $\omega_k$ runs through $\mathbb F_q^*$, the product does the same, consequently the row indexed by $h$ is a permutation of the row indexed by $1$.
\end{proof}

\subsection{Bounds}
In the case of $\mathbb F_{p^m}$, utilizing the bound for the Kloosterman sums reported above,
 the cardinality $ B_{hk}$ as obtained in equation (\ref{geneq}) can be upper bounded as
$$   B_{hk}  = \frac{1}{p^2}  \left[ q   +  \sum_{j=1}^{p-1} \sum_{\ell=1}^{p-1}   e(-jh -\ell k)   \{1+K(\chi_p, j\ell u)\}   \right]  < \frac{1}{p^2}  \left[ q   +  \sum_{j=1}^{p-1} \sum_{\ell=1}^{p-1} |  e(-jh -\ell k)   \{1+K(\chi_p, j\ell u)\} |  \right],   $$
which is bounded by
$$   \frac{1}{p^2}  \left[ q   + (p-1)^2 (1+2\sqrt{q})   \right] = p^{m-2}+  \frac{ (p-1)^2}{p^2}  (1+2\sqrt{p^m}) < p^{m-2}+  1+2\sqrt{p^m} ~~.     $$
Analogously $ B_{hk}$ can be lower bounded by
$  p^{m-2}-  1-2\sqrt{p^m} ~~.     $
Therefore
$$ |B_{hk} -p^{m-2}| < 1+ 2 \sqrt{p^m}     ~~~.  $$

% ********************************************************************

\section{Examples}
Consider the field $\mathbb F_{3^5}$ with generator a primitive polynomial $g(x)=x^5+2 x^3+2 x^2+x+1$ and let $f(z) = \frac{2 \eta^3+\eta}{z}$  with $\eta$ a root of $g(x)$ in $\mathbb F_{3^5}$. The nine values of $B_{hk} $ can be presented in a table where the first line gives the index $k$ and the first  column  gives the index $h$ as elements of $\mathbb F_3$. 

$$  \begin{array}{c|c|c|c|}
                 &  0   &   1  &   2  \\     \hline
             0   &  23   &   29 &   29  \\  \hline
             1   &  29   &   20 &   32  \\   \hline
             2   &  29   &   32  &   20  \\  \hline
       \end{array}
$$
This table clearly shows the symmetries proved in Lemma \ref{lem1} and Theorem \ref{lem2}.
The entry $29$ in the second line of this table can be directly computed using equation (\ref{zeroeq})
$$   B_{10}= \frac{1}{9}  \left[  3^5 -2+  \sum_{j=1}^{2} \sum_{\ell=1}^{2}   e(-j)   K(\chi_3, j\ell (2 \eta^3+\eta))   \right]  =\frac{1}{9} [243-2+20]= 29~~. $$
All values in the table are included in the range $[-5, 59]$.  These loose bounds are due to the small values of $p$ and $m$; for larger values the bounds are tighter.

\vspace{3mm}
\noindent
Consider the field $\mathbb F_{5^5}$ with primitive polynomial generator $g_5(x)=x^5+4 x^4+3 x^3+x^2+2 x+3$ and  let $f(z) = \frac{\theta^3+2\theta^2+3\theta}{z}$ where $\theta$ is a root of $g_5(x)$  in $\mathbb F_{5^5}$. The twenty five values of $ B_{hk} $ are presented in the following table 
%
% with the same conventions of the previous example

$$  \begin{array}{c|c|c|c|c|c|}
                   &      0   &   1  &   2 & 3  & 4 \\     \hline
             0    &  145   &   120  &   120 & 120  & 120 \\     \hline
             1    &  120   &   132  &   108 & 141  & 124 \\     \hline
             2    &  120   &   108  &   124 & 132  &  141 \\     \hline
             3    &  120   &   141  &   132 & 124  & 108 \\     \hline
             4    &  120   &   124  &   141 & 108  & 132 \\     \hline
       \end{array}
$$
The entry $120$ in the second line of this table can be directly computed using equation (\ref{zeroeq}) as
$$   B_{10}= \frac{1}{25}  \left[  5^5 -4+  \sum_{j=1}^{4} \sum_{\ell=1}^{4}   e(-j)   K(\chi_5, j\ell (  \theta^3+2\theta^2+3\theta))   \right]  =\frac{1}{25} [3125-4-121]= 120~~. $$
All values in the table are included in the range $[13, 237]$.

\vspace{3mm}
\noindent
Consider the field $\mathbb F_{4^4}$ with primitive polynomial generator $g_4=x^8+x^7+x^3+x^2+1$ and  let $f(z) = \frac{\zeta_8^7+\zeta_8^2}{z}$ where $\zeta_8 \in \mathbb F_{4^4}$ is a root of $g_4(x)$  in $\mathbb F_{4^4}$. 
The first line and column of the table contain the elements of $\mathbb F_{4}$ which play the role of indices in Theorem \ref{theo1}  
$$  \begin{array}{c|c|c|c|c|c|}
                   &    0   &   1  &   \zeta_2 & \zeta_2^2   \\     \hline
             0    &  19   &   15  &   15 & 15  \\     \hline
             1    &  15   &   12  &    21 & 16   \\     \hline
   \zeta_2    &  15   &   21  &   16 & 12  \\     \hline
  \zeta_2^2 &  15   &   16  &   12 & 21   \\     \hline
       \end{array}
$$
where $\zeta_2=\zeta_8^{17}$  is a generator of $\mathbb F_{4}^*$ as a subgroup of $\mathbb F_{4^4}^*$.

\noindent
The entry $15$ in the second line of this table can be directly computed using equation (\ref{zeroeq1}) as
$$   B_{01}=   4^4 -3 \cdot 64-3 \cdot 64 + 143= 15~~. $$

\section{Applications}
The generation of random numbers  \cite{knuth}  is an important endeavor in many theoretical and applied sciences including areas ranging from bioengineering to positioning systems, from Montecarlo techniques in numerical computations, to system simulations, and to cryptography. In this special domain, random numbers are 
crucial in the management of keys of both public and private key crypto-systems. 
For instance, in RSA based systems the random generation of large primes is instrumental to the security of the scheme. In the definition of the initial state of stream ciphers \cite{rueppel} based on linear feedback shift registers \cite{golomb} or block ciphers like AES,
the common secret key should be a secret random number to avoid direct attacks based on guessing unwisely chosen keys. In these very demanding applications,
the involution described in the previous section can be used as a component of a mechanism that generates secret pseudorandom numbers that are actually uniformly distributed over a finite set of keys. While ways of exploiting these functions for pseudorandom number generators can already be found in the literature, for instance in  \cite{nieder,shp,winter}, we present below a direct application of the theorems proved in the previous sections.

\noindent 
Abstractly, the problem is to generate a pseudorandom object~ $\mathbf r$, usually a number, in a
finite set $\mathfrak I$. %, a generation that should not suffer of any bias due to the generation method or to the operator. \\
Suppose  we want to produce uniformly distributed random numbers in the set $\{1,  \ldots, q-1 \}$, where $q$ is a prime power. Then we could pick  $\gamma$ in some extension $\mathbb F_{q^m}$, and generate a number
$\mathbf R \in \mathbb F_q$ as
$$ \mathbf R = Tr_q\left(\frac{u}{\gamma}\right)   ~~~u \in   \mathbb F_{q^m}.  $$
In view of the above theory, we should consider a $\gamma\neq 0$ with $0$-trace and disregard the output if $Tr_q\left(\frac{u}{\gamma}\right)=0$. 

\noindent 
Moreover the parameter $u$ should be thought of as an unknown, i.e. it is chosen by some device that is not under the control of the tester. Therefore, even knowing $\gamma$, it is impossible to predict the value of $\mathbf R$ with a probability different from $\frac{1}{q-1}$. Further, the probability of guessing $u$ would be $\frac{1}{q}$ since we have
$$    Tr\left(\frac{u_1}{\gamma}\right)  =Tr\left(\frac{u}{\gamma}\right)   ~~\Rightarrow ~~Tr\left(\frac{u_1-u}{\gamma}\right) =0  $$
and this equation implies  $u_1=u+\beta_0 \gamma$ with $\beta_0$ any element of $ \mathbb F_{q^m}$ of $0$-trace. Therefore, the probability of choosing a valid $u$ at random is $\frac{1}{q}$.

%a method can be the following:
%\begin{enumerate}
 % \item Let $p$ be a prime, and choose an integer $m> 1$  such that
 %   $p^m$ is large.
 % \item Take a primitive element $u_0$ of $\mathbb F_{p^m}$.
 % \item Take an element $k_0$ of $\mathbb F_{p^m}$, different from $0$ and $1$, and whose trace is %$0$. %, if not
  % subtract from $k_0$ an element of trace equal to $Tr(k_0)$.
 % \item Obtain  $\mathbf R$, the random number, as the first non-zero trace value in the sequence 
%$$    Tr_p\left(\frac{u_0}{k_0} \right), Tr_p\left(\frac{u_0^2}{k_0} \right), \ldots, %%Tr_p\left(\frac{u_0^s}{k_0} \right), ~~,     
%$$
%\end{enumerate}

\noindent
%In view of Theorem \ref{lem2}, $\mathbf R$ is an instance of a random variable  uniformly %distributed in the 
%set $\{1, 2, \ldots, p-1 \}$. 
As a consequence of the {bounds} on cardinalities $B_{hk}$, it is possible to drop the conditions on the trace being equal  or not to $0$ as above, and still have a close to uniform distribution, provided the parameters of the field are appropriately chosen.
%(as a consequence of the {bounds} on cardinalities $\mathcal B_{hk}$).

\noindent
If it is desired to obtain an instance $\mathbf r$ of a random variable uniformly distributed in a different set $\mathfrak I$, a one-to-one mapping $\Phi$ between $\{1, 2, \ldots, q-1 \}$ and $\mathfrak I$ should be considered leading to  $\mathbf r$ as $\Phi(\mathbf R)$.  \\
If the cardinality $w$ of $\mathfrak I$ is not $q-1$ above, then consider a prime of the form $Q=2\mu w+1$
(which always exists by a theorem of Dirichlet's), and obtain $\mathbf R$ as described above. 
%Then take the residue $\mathbf S=\mathbf R \bmod w$ which turns out to be an instance of an equally %distributed random variable over $\{0,1, \ldots, w-1 \}$. 
The residue  $\mathbf S=\mathbf R \bmod w$ is then taken, and is found to be an instance of an equally distributed random variable over $\{0,1, \ldots, w-1 \}$.
Finally, a one-to-one mapping $\Psi$ between the set $\{0,1,\ldots,w-1\}$ and $\mathfrak I$ is defined, obtaining $\mathbf r \in \mathfrak I$ as $\Psi(\mathbf S)$.

\section{Acknowledgements}
Research was partially supported by COST Action IC1306 and Swiss National Science Foundation grant
No. 149716.

%\section{Conclusions}
%A scheme for generating pseudorandom numbers with provable uniform probability distribution  is %proposed.
%The generating mechanism is based on properties of the partitions of extension finite fields in %subsets of elements with the same trace.
% In particular the function $b=Tr(\frac{u}{z})$ has been used as a map from $\mathbb F_{q^r}$  onto $\mathbb F_q$, which is many-to-one, thus certainly not invertible. Furthermore, given  $z$, a priori it is not possible to predict $b$. 
% If we want an instance $\mathbf r$ of a random variable uniformly distributed in a different set $\mathfrak I$, we plainly consider a one-to-one mapping $\Phi$ between $\{1, 2, \ldots, q-1 \}$ and $\mathfrak I$, and obtain $\mathbf r$ as $\Phi(\mathbf R)$.  \\
%If the cardinality $w$ of $\mathfrak I$ is not $q-1$, then consider a prime of the form $Q=2\mu w+1$
%(which always exists by a theorem of Dirichlet's), and produce $\mathbf R$ as described above. 
%Then take the residue $\mathbf S=\mathbf R \bmod w$ which turns out to be an instance of an equally distributed random variable over $\{0,1, \ldots, w-1 \}$.  \\

%\vspace{3mm}
%\noindent
%
%  \textcolor{blue}{Is the case to mention the open problem of bound computations for the trace function into $\mathbb F_{p^m}$?? } 

% ******************************************************************

\end{document}